\documentclass[aps,pra,reprint,superscriptaddress,nofootinbib,longbibliography]{revtex4-2}

\usepackage{amsmath}
\usepackage{amssymb}
\usepackage{amsthm}
\usepackage{bm}
\usepackage{graphicx}
\usepackage{hyperref}
\usepackage{orcidlink}

\newtheorem{lemma}{Lemma}
\newtheorem{proposition}{Proposition}
\newtheorem{corollary}{Corollary}

\newcommand{\bra}[1]{\langle #1 |}
\newcommand{\ket}[1]{| #1 \rangle}
\newcommand{\braket}[2]{\langle #1 | #2 \rangle}
\newcommand{\Tr}{\operatorname{Tr}}
\newcommand{\eps}{\varepsilon}
\newcommand{\HH}{\mathcal{H}}

\newcommand{\E}{\mathbb{E}}
\newcommand{\PP}{\mathbb{P}}

\begin{document}

\title{Double-Exponential Quasi-Orthogonality:
The Geometry of Decoherence}

\author{Karl Svozil\,\orcidlink{0000-0001-6554-2802}}
\affiliation{Institute for Theoretical Physics, TU Wien,
Wiedner Hauptstrasse 8--10/136, A-1040 Vienna, Austria}
\email{karl.svozil@tuwien.ac.at}

\date{\today}

\begin{abstract}
A composite system of $N$ local $q$-level factors has dimension $D=q^N$.
Although at most $D$ vectors can be exactly orthogonal, a fixed squared-overlap
tolerance $\eps$ permits $M_\eps(D)\gtrsim\exp(c_\eps D)$ mutually
quasi-orthogonal directions.  Consequently
$M_\eps(q^N)\gtrsim\exp(c_\eps q^N)$: the dimension grows exponentially in
$N$, but its quasi-orthogonal capacity grows doubly exponentially.
L\'evy's lemma explains the accompanying concentration of regular observables,
and Johnson--Lindenstrauss scaling gives the complementary finite-set account
of exponential capacity.  For Haar-random pure states the sharper exact law
is
$\mathbb P(|\langle\phi|\psi\rangle|^2\geq\eps)=(1-\eps)^{D-1}$,
while typical Fubini--Study angles lie within $O(D^{-1/2})$ of $\pi/2$:
capacity explodes as angular structure homogenizes.  We turn these facts into
simultaneous overlap and trace-distance bounds for finite decoherence branch
families, including mixed environments and collective weak coherences.  The
results are conditional on typical relative environmental dynamics.  They
neither select a pointer basis nor identify coherence suppression with
readable or redundant records.
\end{abstract}

\maketitle

\section{Introduction}

The quantum measurement problem includes several logically distinct
questions: which observables acquire a preferred status, why outcome
statistics obey the Born rule~\cite{gleason,Wright_2019} and why individual observations appear to
have definite outcomes although the global state may remain a coherent
superposition~\cite{everett,schlosshauer-2007,schlosshauer-2019}.
Seventeen years after introducing his famous cat thought
experiment~\cite{schrodinger-gwsidqm1}, Schr\"odinger returned to this tension
in a 1952 Dublin colloquium.  If the wave equation applies universally, he
observed, a naive reading appears to turn macroscopic surroundings into a
``quagmire'' or ``featureless jelly'' and even seems to threaten the stability
of an unobserved wristwatch in a drawer~\cite{schroedinger-interpretation}.

Environmental
decoherence addresses an important operational part of this problem.  It
explains why interference between certain alternatives becomes inaccessible
to measurements on a system once information about those alternatives has
dispersed into its environment.  Decoherence does not, by itself, select one
member of the resulting global superposition as the unique outcome.

The modern operational response is not that the global superposition
disappears.  Rather, environmental
correlations make its relative phases inaccessible to the restricted
observables by which macroscopic systems are ordinarily interrogated.

Standard decoherence theory supplies the dynamical framework: the
system--environment interaction selects stable pointer observables and
correlates their alternatives with conditional environmental
states~\cite{joos1985emergence,RevModPhys.75.715,joos2003decoherence,schlosshauer-2007}.
High dimension alone does not guarantee decoherence: a relative conditional
unitary may, for example, remain the identity on the entire environmental
space.  The question addressed here is narrower: once an explicit typicality
assumption for the conditional environmental states or relative unitaries has
been specified, what suppression follows from high-dimensional geometry
alone?  Unit vectors sampled independently from a high-dimensional complex
Hilbert space normally have very small overlaps.  This is an elementary
instance of concentration of
measure~\cite{tao2012topics,vershynin2018high,ledoux2001concentration,milman1986asymptotic}.

The word \emph{independently} is essential.  Actual environmental branch
states arise from the same initial state under correlated conditional
dynamics.  Their relative unitary need not be Haar distributed, and a large
energy shell does not make it so.  We therefore formulate the geometric
results as conditional bounds, identify the precise dynamical object to which
they would apply, and distinguish three notions that are easily conflated:
small reduced-system coherences, distinguishable conditional environmental
states, and redundant classical records.

The individual geometric identities used below are standard.  The purpose of
this paper is to combine them into operational trace-distance estimates for
finite families of conditional environmental states, to extend the
mean-square analysis to mixed environments, and to delineate the dynamical
and information-theoretic assumptions needed before these kinematic estimates
can be interpreted as physical decoherence claims.  The organizing idea is
that tensor composition and quasi-orthogonal packing are separate
amplifications: the first makes dimension exponential in subsystem number,
and the second makes the number of available almost-orthogonal directions
exponential in that already enlarged dimension.  L\'evy concentration
explains why this near-orthogonality is typical, while the
Johnson--Lindenstrauss scaling explains, from the finite-set side, why the
number of simultaneously representable directions is exponential in
dimension.

\section{Two geometric amplifications}
\label{sec:two-amplifications}

Consider $N$ local factors of dimension $q\geq2$.  Their finite tensor
product has dimension
\begin{equation}
 D=q^N.
 \label{eq:tensor-dimension}
\end{equation}
This is the familiar exponential growth caused by composition.  It is not,
however, the only large-number effect.  Exact orthogonality allows no more
than $D$ mutually orthogonal vectors, but if the requirement is relaxed to
\begin{equation}
 |\langle\psi_i|\psi_j\rangle|^2\leq\eps\qquad(i\neq j),
 \label{eq:quasiorthogonality-preview}
\end{equation}
then projective Hilbert space admits families of cardinality
$M_\eps(D)\gtrsim\exp(c_\eps D)$ for every fixed $0<\eps<1$; the explicit
random-coding bound is given in Eq.~\eqref{eq:packing}.  Substitution of
Eq.~\eqref{eq:tensor-dimension} gives
\begin{equation}
 M_\eps(q^N)\gtrsim\exp(c_\eps q^N),
 \label{eq:double-explosion}
\end{equation}
which is doubly exponential in the number of factors.  The first exponential
therefore comes from tensor composition, while the second comes from allowing
almost, rather than exact, orthogonality.

The same phenomenon has a complementary angular description.  For two
independent Haar directions, the typical overlap amplitude is
$O(D^{-1/2})$, so their Fubini--Study angle satisfies
\begin{equation}
 \theta_{\rm FS}=\arccos|\langle\phi|\psi\rangle|
 =\frac{\pi}{2}-O(D^{-1/2}).
 \label{eq:angular-homogenization}
\end{equation}
As $D$ grows, an enormous number of directions becomes available, yet most
pairwise angles become nearly indistinguishable from a right angle.  We call
this \emph{angular homogenization}, or a concentration-induced loss of
angular resolution.  It is not a literal loss of quantum information and it
does not increase the linear dimension beyond $D$; it describes the
coarsening of typical pairwise angular structure at fixed experimental
resolution.

The two amplifications concern different variables.  Tensor composition
changes the ambient dimension as a function of $N$.  Concentration and
packing describe the distribution and number of directions inside that
ambient space.  The next two sections develop these statements through
L\'evy's and Johnson--Lindenstrauss lemmas before the exact overlap law is
used to obtain the sharper operational estimates.

\section{L\'evy concentration and angular homogenization}
\label{sec:levy}

Identify $\mathbb C^D$ with $\mathbb R^{2D}$ and write
\begin{equation}
 \mathbb S^{2D-1}=\{\ket{\psi}\in\mathbb C^D:
 \langle\psi|\psi\rangle=1\}.
\end{equation}
Uniform measure on this sphere is the pure-state measure induced by Haar
measure on $U(D)$.  A real-valued function $f$ on the sphere is
$L$-Lipschitz with respect to chordal distance if
\begin{equation}
 |f(\psi)-f(\varphi)|\leq L\|\psi-\varphi\|.
 \label{eq:lipschitz}
\end{equation}
In a convenient normalization, L\'evy's lemma states that
\cite{milman1986asymptotic,ledoux2001concentration,Hayden2006}
\begin{equation}
 \PP\{|f-\E f|\geq\delta\}
 \leq 2\exp\!\left[-\frac{(2D-1)\delta^2}{9\pi^3L^2}\right].
 \label{eq:levy}
\end{equation}
Versions centered at a median, and versions with different universal
constants, express the same dimension-dependent concentration.  The content
is that every sufficiently regular scalar probe is almost constant over
almost all of a high-dimensional sphere.

To connect this general theorem to quasi-orthogonality, fix a unit vector
$\ket{\phi}$ and use
\begin{equation}
 f_\phi(\psi)=|\langle\phi|\psi\rangle|^2.
 \label{eq:overlap-function}
\end{equation}
Unitary invariance and normalization give
\begin{equation}
 \E f_\phi=\frac1D.
 \label{eq:haar-overlap-mean-preview}
\end{equation}
Moreover,
\begin{align}
 |f_\phi(\psi)-f_\phi(\varphi)|
 &\leq
 \bigl(|\langle\phi|\psi\rangle|+|\langle\phi|\varphi\rangle|\bigr)
 |\langle\phi|\psi-\varphi\rangle| \nonumber\\
 &\leq2\|\psi-\varphi\|,
\end{align}
so $L\leq2$.  Equation~\eqref{eq:levy} therefore yields
\begin{equation}
 \PP\!\left\{
 \left||\langle\phi|\psi\rangle|^2-\frac1D\right|\geq\delta
 \right\}
 \leq2\exp\!\left[-\frac{(2D-1)\delta^2}{36\pi^3}\right].
 \label{eq:levy-overlap}
\end{equation}
Here concentration becomes quasi-orthogonality because the value about which
the function concentrates is itself small.  The sphere is overwhelmingly
equatorial relative to any fixed direction: states with appreciable overlap
occupy small spherical caps.

This application also shows exactly what the generic theorem does and does
not provide.  For a fixed excess threshold $\delta$, it gives an
exponentially small exceptional measure.  At the natural squared-overlap
scale $\delta=O(D^{-1})$, however, the displayed generic bound becomes weak.
The exact beta distribution in Lemma~\ref{lem:haar-overlap} resolves that
scale and gives the sharper tail $\exp[-O(D\eps)]$.  Thus L\'evy's lemma is
the general geometric mechanism; the beta law is the specialized tool used
for the quantitative decoherence and packing bounds below.

The corresponding projective statement is angular homogenization.  Since
the typical overlap amplitude is $O(D^{-1/2})$, the Fubini--Study angle
obeys Eq.~\eqref{eq:angular-homogenization}.  Concentration therefore has two
faces: almost all directions are almost orthogonal to a specified direction,
and most angular comparisons become compressed into a narrow band about
$\pi/2$.  This is the precise geometric sense in which angular information
becomes poorly resolved; no global state information is destroyed by the
unitary theory.

\section{Johnson--Lindenstrauss scaling and finite configurations}
\label{sec:jl}

L\'evy's lemma is a statement about measure on a sphere.  The
Johnson--Lindenstrauss lemma is complementary: it concerns the faithful
representation of a specified finite metric configuration.  For
$0<\delta<1$ and a set of $M$ points in Euclidean space, there exists a map
into $k$ dimensions, with
\begin{equation}
 k\leq C\delta^{-2}\log M,
 \label{eq:jl-dimension}
\end{equation}
that preserves all squared pairwise distances within multiplicative factors
$1\pm\delta$~\cite{johnson-1984}.  The universal constant $C$ and refinements
of the dependence on $\delta$ are immaterial here; the decisive feature is
the logarithmic dependence on $M$.

Apply this result to the $M$ standard basis vectors of $\mathbb R^M$ and
normalize their projected images.  Because their original norms are one and
their mutual squared distances are two, preservation of norms and distances
implies pairwise inner products of magnitude $O(\delta)$.  Real vectors are
a subset of complex Hilbert space, so the construction also supplies a
complex spherical code.  Reading Eq.~\eqref{eq:jl-dimension} in the opposite
direction gives
\begin{equation}
 M\geq\exp(c\delta^2D)
 \label{eq:jl-inverse}
\end{equation}
for a suitable fixed distortion and sufficiently large $D$.  If the desired
squared-overlap tolerance is denoted by $\eps$, then
$\eps=O(\delta^2)$ and Eq.~\eqref{eq:jl-inverse} has the same
$\exp(c\eps D)$ scale as the direct Haar random-coding bound.

The two results should nevertheless not be identified.  Johnson--Lindenstrauss
preserves the metric of an arbitrary finite point set after dimensional
reduction; the packing argument asks how many normalized vectors can be
placed in a fixed Hilbert space subject to an inner-product constraint.
Using the standard basis connects them constructively, while the exact
complex-Haar calculation below supplies the natural probability law and
better constants for the present problem.  Together, L\'evy concentration
and Johnson--Lindenstrauss scaling distinguish typicality from capacity:
the former says that near-orthogonality occupies almost all directions, and
the latter explains why an exponential number of finite directions can be
accommodated at fixed resolution.

\section{Decoherence factors and an operational bound}
\label{sec:operational}

Let $\{\ket{s_i}\}_{i=1}^{m}$ be orthonormal pointer alternatives spanning
the support of the system state, selected by a measurement-like interaction,
and let $\sum_i|c_i|^2=1$.  For a pure initial environment $\ket{E_0}$,
the joint evolution has the form $\ket{E_i(t)}=U_i(t)\ket{E_0}$, with
\begin{equation}
 \left(\sum_{i=1}^{m}c_i\ket{s_i}\right)\ket{E_0}
 \longmapsto
 \ket{\Psi(t)}=
 \sum_{i=1}^{m}c_i\ket{s_i}\ket{E_i(t)}
.
 \label{eq:branching}
\end{equation}
Tracing out the environment gives
\begin{equation}
 \rho_S(t)=
 \sum_{i,j=1}^{m}c_ic_j^*\gamma_{ji}(t)
 \ket{s_i}\bra{s_j},
 \qquad
 \gamma_{ji}(t)=\braket{E_j(t)}{E_i(t)}.
 \label{eq:reduced}
\end{equation}
Let $\Delta$ denote complete dephasing in the pointer basis,
\begin{equation}
 \Delta(\rho_S)=\sum_i\ket{s_i}\bra{s_i}\rho_S
 \ket{s_i}\bra{s_i}.
\end{equation}
Thus $\Delta$ is trace preserving on the relevant support.  All coherence
statements below are relative to this dynamically selected pointer basis;
they do not define a basis-independent quantity called coherence.
The trace distance
$D_{\rm tr}(\rho,\sigma)=\frac12\|\rho-\sigma\|_1$ has a direct operational
meaning: it controls the maximum change of outcome probabilities over all
measurements.  The environmental overlaps give the following elementary
bound.

\begin{proposition}[Reduced-state coherence bound]
\label{prop:trace-bound}
For the state in Eq.~\eqref{eq:reduced},
\begin{align}
 D_{\rm tr}\!\left(\rho_S,\Delta(\rho_S)\right)
 &\leq
 \min\!\left\{1,
 \frac12\sum_{i\neq j}|c_i c_j|\,|\gamma_{ji}|\right\}.
 \label{eq:trace-bound-general}
\end{align}
It also obeys the Hilbert--Schmidt bound
\begin{equation}
 D_{\rm tr}\!\left(\rho_S,\Delta(\rho_S)\right)
 \leq \min\!\left\{1,\frac12\sqrt{mZ}\right\},
 \label{eq:trace-bound-hs}
\end{equation}
where
\begin{equation}
 Z=\sum_{i\neq j}|c_i|^2|c_j|^2|\gamma_{ji}|^2.
 \label{eq:collective-coherence}
\end{equation}
If $|\gamma_{ji}|\leq\delta$ for all $i\neq j$, then
\begin{equation}
\begin{aligned}
 D_{\rm tr}\!\left(\rho_S,\Delta(\rho_S)\right)
 &\leq
 \min\!\left\{1,\frac{\delta}{2}
 \left[\left(\sum_i|c_i|\right)^2-1\right]
 \right\} \\
 &\leq \min\!\left\{1,\frac{m-1}{2}\delta\right\}.
 \label{eq:trace-bound-uniform}
\end{aligned}
\end{equation}
For two alternatives the first inequality is an equality:
$D_{\rm tr}=|c_1c_2|\,|\gamma_{12}|$.
\end{proposition}

\begin{proof}
Writing
\begin{equation}
 X=\rho_S-\Delta(\rho_S)
 =\sum_{i\neq j}c_ic_j^*\gamma_{ji}\ket{s_i}\bra{s_j},
\end{equation}
the triangle inequality, the trivial bound $D_{\rm tr}\leq1$, and
$\|\ket{s_i}\bra{s_j}\|_1=1$ give
Eq.~\eqref{eq:trace-bound-general}.  Since the difference has rank at most
$m$, $\|X\|_1\leq\sqrt{m}\|X\|_2$ and direct evaluation of
$\|X\|_2^2=\Tr(X^\dagger X)$ give
Eq.~\eqref{eq:trace-bound-hs}.  The first bound in
Eq.~\eqref{eq:trace-bound-uniform} follows by taking the maximum overlap out
of the sum.  The second follows from
$\sum_i|c_i|\leq\sqrt{m}$ and $\sum_i|c_i|^2=1$.
For $m=2$, the two nonzero singular values of the off-diagonal difference are
both $|c_1c_2\gamma_{12}|$.
\end{proof}

Pairwise quasi-orthogonality is therefore useful only together with control
of the number and amplitudes of the alternatives.  A tiny entrywise bound can
fail to control a large Gram matrix if sufficiently many off-diagonal entries
accumulate.

\section{Haar geometry and simultaneous branch bounds}
\label{sec:haar}

\subsection{Exact overlap distribution}

\begin{lemma}[Complex Haar overlap]
\label{lem:haar-overlap}
Let $\ket{\psi},\ket{\phi}\in\mathbb C^D$, $D\geq2$, be independent
Haar-random unit vectors.  Then
\begin{equation}
 T=|\langle\phi|\psi\rangle|^2
 \sim\operatorname{Beta}(1,D-1),
 \label{eq:beta}
\end{equation}
and, for $0\leq\eps\leq1$,
\begin{equation}
 \PP(T\geq\eps)=(1-\eps)^{D-1}
 \leq e^{-(D-1)\eps}.
 \label{eq:tail}
\end{equation}
In particular, $\E T=1/D$.
\end{lemma}

\begin{proof}
By unitary invariance fix $\ket{\phi}=\ket{e_1}$.  A Haar-random unit
vector can be written as $Z/\|Z\|$, where the $Z_k$ are independent standard
complex Gaussian variables.  The variables $X_k=|Z_k|^2$ are independent
exponentials of mean one.  Hence
\begin{equation}
 T=\frac{X_1}{X_1+\cdots+X_D}
 \sim\operatorname{Beta}(1,D-1).
\end{equation}
Integration of its density $(D-1)(1-t)^{D-2}$ gives
Eq.~\eqref{eq:tail}; its mean is $1/D$.
\end{proof}

Thus the typical squared overlap is of order $D^{-1}$ and the typical
amplitude is of order $D^{-1/2}$.  The exact distribution is more informative
than a generic application of L\'evy's lemma because it retains the precise
dependence on the overlap threshold.  Equivalently, the Fubini--Study angle
is typically $\pi/2-O(D^{-1/2})$: concentration produces not only small
overlaps but also the angular homogenization described in
Sec.~\ref{sec:two-amplifications}.

\subsection{A finite-family theorem}

The quantity relevant for a measurement with several alternatives is the
largest pairwise overlap, rather than the overlap of one pair.

\begin{proposition}[Simultaneous overlap and trace-distance bound]
\label{prop:simultaneous}
Let $m\geq2$, let $\ket{E_1},\ldots,\ket{E_m}$ be independent Haar-random
unit vectors in $\mathbb C^D$, and let $0<\eta<1$.  Define
\begin{equation}
 K=\binom{m}{2},\,
 x_\eta=\frac{\log(K/\eta)}{D-1},\,
 \eps_\eta=1-e^{-x_\eta}.
 \label{eq:epsiloneta}
\end{equation}
With probability at least $1-\eta$,
\begin{equation}
 \max_{i\neq j}|\langle E_j|E_i\rangle|
 \leq\sqrt{\eps_\eta}
 \leq
 \sqrt{\frac{\log[\binom{m}{2}/\eta]}{D-1}}.
 \label{eq:max-overlap}
\end{equation}
Consequently, the branched state in Eq.~\eqref{eq:branching} obeys, with the
same probability,
\begin{align}
 D_{\rm tr}\!\left(\rho_S,\Delta(\rho_S)\right)
 &\leq
 \min\!\left\{1,
 \frac12\left[\left(\sum_i|c_i|\right)^2-1\right]
 \sqrt{\eps_\eta}\right\}
 \nonumber\\
 &\leq
 \min\!\left\{1,
 \frac{m-1}{2}
 \sqrt{\frac{\log[\binom{m}{2}/\eta]}{D-1}}\right\}.
 \label{eq:prob-trace-bound}
\end{align}
\end{proposition}

\begin{proof}
Lemma~\ref{lem:haar-overlap} and a union bound give
\begin{equation}
 \PP\!\left(\max_{i<j}|\langle E_j|E_i\rangle|^2>\eps\right)
 \leq\binom{m}{2}(1-\eps)^{D-1}.
\end{equation}
Equation~\eqref{eq:epsiloneta} makes the right-hand side equal to $\eta$.
The second inequality in Eq.~\eqref{eq:max-overlap} follows from
$1-e^{-x}\leq x$.  Proposition~\ref{prop:trace-bound} completes the proof.
\end{proof}

\begin{corollary}[Pure-Haar collective bound]
\label{cor:pure-haar-collective}
Under the hypotheses of Proposition~\ref{prop:simultaneous}, with probability
at least $1-\eta$,
\begin{equation}
 D_{\rm tr}\!\left(\rho_S,\Delta(\rho_S)\right)
 \leq
 \min\!\left\{1,
 \frac12\sqrt{
 \frac{m}{D\eta}
 \left(1-\sum_i|c_i|^4\right)}\right\}.
 \label{eq:pure-haar-collective}
\end{equation}
\end{corollary}

\begin{proof}
For $i\neq j$, Lemma~\ref{lem:haar-overlap} gives
$\E|\gamma_{ji}|^2=D^{-1}$.  Therefore the quantity $Z$ in
Eq.~\eqref{eq:collective-coherence} satisfies
\begin{equation}
 \E Z=\frac{1}{D}\left(1-\sum_i|c_i|^4\right).
\end{equation}
Markov's inequality gives $Z\leq\E Z/\eta$ with probability at least
$1-\eta$.  Substitution into Eq.~\eqref{eq:trace-bound-hs} proves the claim.
\end{proof}

This collective estimate is complementary to Proposition~\ref{prop:simultaneous}:
it can improve the branch-number scaling for broadly distributed weak
coherences, whereas its confidence dependence is only $\eta^{-1/2}$ rather
than logarithmic.

For fixed $m$ and fixed confidence, the operational error in this bound is
$O(D^{-1/2})$.  If $m$ grows with $D$, the accumulation factor in
Eq.~\eqref{eq:prob-trace-bound} must be retained.  This qualification is
important in discussions of very large branch families.  For comparable
amplitudes the displayed bound scales as
$m\sqrt{\log(m)/D}$ and therefore becomes nontrivial in a considerably more
restricted regime than the pairwise packing bound below.  The union bound
does not require mutual independence of the pair-overlap events.  The same
conclusion holds for any joint ensemble whose pairwise marginals have the
Haar overlap distribution of Lemma~\ref{lem:haar-overlap}.

The tradeoff can also be read directly from Eq.~\eqref{eq:epsiloneta}.  At
fixed $D$ and failure probability $\eta$, increasing the number of branches
$m$ increases the overlap threshold $\eps_\eta$ that can be guaranteed.
Equivalently, keeping both $\eta$ and the allowed overlap fixed while
increasing $m$ requires a larger effective dimension $D$.

\subsection{Packing is not classical memory capacity}

The same union-bound calculation gives a random-coding statement.  For
$0\leq\eps<1$, let
$M_\eps(D)$ denote the largest cardinality of a family of unit vectors in
$\mathbb C^D$ satisfying
$|\langle\psi_i|\psi_j\rangle|^2\leq\eps$ for every $i\neq j$.  Then
\begin{equation}
 M_\eps(D)\geq
 \left\lfloor
 \exp\!\left[\frac{D-1}{2}\bigl(-\log(1-\eps)\bigr)\right]
 \right\rfloor.
 \label{eq:packing}
\end{equation}
Indeed, for this number of independent Haar samples the union-bound
probability of a violating pair is strictly below one.  For fixed small
$\eps$, Eq.~\eqref{eq:packing} is exponential in $\eps D$.  In particular,
for $D=q^N$ it yields
\begin{equation}
 M_\eps(q^N)\geq
 \left\lfloor\exp\!\left[
 \frac{q^N-1}{2}\bigl(-\log(1-\eps)\bigr)
 \right]\right\rfloor,
 \label{eq:packing-tensor}
\end{equation}
a doubly exponential growth with $N$.  Likewise, for any fixed pair the
failure probability in Eq.~\eqref{eq:tail} is bounded by
$\exp[-\eps(q^N-1)]$, which is doubly exponentially small in $N$.
This second amplification is present only after the tolerance is fixed: the
linear dimension remains $q^N$, and no more than $q^N$ states can be exactly
orthogonal.

This direct complex-Haar construction realizes the same exponential scale as
the Johnson--Lindenstrauss argument in Sec.~\ref{sec:jl}.  Here, however, the
constant follows from the exact overlap tail, and the result directly
controls the squared inner products relevant to environmental branches.

This spherical-code packing bound is not a capacity for reliably readable
classical records.  When $M>D$, the states are linearly dependent and cannot
be perfectly discriminated by a single-copy measurement.  The Welch bound
further gives
\begin{equation}
 \max_{i\neq j}|\langle\psi_i|\psi_j\rangle|^2
 \geq \frac{M-D}{D(M-1)}
 \qquad (M>D),
 \label{eq:welch}
\end{equation}
Equation~\eqref{eq:welch} is a geometric restriction, not by itself a bound
on a specified decoder~\cite{welch1974coherence}.  Approximate discrimination
requires a prior distribution, a measurement, and an error criterion.  For a
uniform ensemble of $M$ pure states in dimension $D$, the Holevo bound
$\chi\leq\log D$, combined with Fano's inequality, gives a nonzero lower bound
on the average decoding error when $\log M-\log D$ is sufficiently large
\cite{Holevo1973,nielsen-book10}.  Equation~\eqref{eq:packing} is therefore
best viewed as a statement about projective Hilbert-space geometry.  The
operationally relevant decoherence statement is instead
Proposition~\ref{prop:simultaneous} for a specified family of alternatives.

\section{Mixed environments: coherence is not a record}
\label{sec:mixed}

This section isolates a counterintuitive but important distinction.  Greater
initial mixedness can reduce the Haar-averaged decoherence factor while
simultaneously reducing, or even eliminating, the environment's capacity to
encode a readable record.  These are different operational quantities, not a
paradox.

For an initially mixed environmental state $\rho_E$, a controlled unitary
produces reduced-system coherence factors
\begin{equation}
\begin{aligned}
 \gamma_{ji}(t)
 &=\Tr\!\left[U_i(t)\rho_EU_j(t)^\dagger\right]
 =\Tr\!\left[\rho_E W_{ji}(t)\right],
 \\
 W_{ji}(t)&=U_j(t)^\dagger U_i(t).
 \label{eq:mixedgamma}
\end{aligned}
\end{equation}
The relevant random object is the \emph{relative} conditional unitary
$W_{ji}(t)$, not two independently postulated purifications.

\begin{lemma}[Haar relative unitary]
\label{lem:mixed}
Let $W$ be Haar-random on a $D$-dimensional Hilbert space and let $\rho$ be a
fixed density operator, chosen independently of $W$, on that space.  Then
\begin{equation}
 \E_W|\Tr(\rho W)|^2=\frac{\Tr(\rho^2)}{D}.
 \label{eq:mixedhaar}
\end{equation}
\end{lemma}

\begin{proof}
In a fixed basis, Haar invariance gives the second-moment identity
\begin{equation}
 \E_W[W_{ab}W_{cd}^*]=\frac{1}{D}\delta_{ac}\delta_{bd}.
\end{equation}
Expanding $|\Tr(\rho W)|^2$ and applying this identity yields
$D^{-1}\sum_{a,b}|\rho_{ab}|^2=D^{-1}\Tr(\rho^2)$.
\end{proof}

The identity requires only the first unitary-design moment, equivalently the
matrix-entry correlation $\E[W_{ab}W_{cd}^*]$.  It immediately yields a
simultaneous mixed-state statement.

\begin{proposition}[Mixed-environment simultaneous bounds]
\label{prop:mixed-simultaneous}
Let $\HH_D$ be a fixed $D$-dimensional environmental sector, let $\rho_E$ be
a density operator on $\HH_D$, and let the normalized coefficient vector
$(c_i)$ be fixed independently of a specified ensemble of relative unitaries
$W_{ji}$ acting on $\HH_D$.  Suppose that, for every $i<j$, the marginal
ensemble of $W_{ji}$ reproduces the matrix-entry correlation
$\E[W_{ab}W_{cd}^*]=D^{-1}\delta_{ac}\delta_{bd}$ on $\HH_D$.  The
probability refers to the specified ensemble,
such as Hamiltonian disorder, circuit randomness, conditional-perturbation
randomness, or a prescribed random-time sampling procedure.  No mutual
independence of the different $W_{ji}$ is required.  Put
$K=\binom{m}{2}$ and let $0<\eta<1$.  For every $\delta>0$,
\begin{equation}
 \PP\!\left\{\max_{i<j}|\gamma_{ji}|\geq\delta\right\}
 \leq K\frac{\Tr(\rho_E^2)}{D\delta^2},
 \label{eq:mixed-union-tail}
\end{equation}
and hence, with probability at least $1-\eta$,
\begin{equation}
 \max_{i<j}|\gamma_{ji}|
 \leq
 \min\!\left\{1,
 \sqrt{\frac{K\Tr(\rho_E^2)}{D\eta}}\right\}.
 \label{eq:mixed-max-bound}
\end{equation}
Moreover, with probability at least $1-\eta$,
\begin{equation}
\begin{split}
 D_{\rm tr}\!\left(\rho_S,\Delta(\rho_S)\right)
 \leq \min\!\left\{1,
 \frac12\sqrt{
 \frac{m\Tr(\rho_E^2)}{D\eta}
 \left(1-\sum_i|c_i|^4\right)}\right\}.
 \label{eq:mixed-trace-tail}
\end{split}
\end{equation}
\end{proposition}

\begin{proof}
Markov's inequality and Lemma~\ref{lem:mixed} give
$\PP\{|\gamma_{ji}|\geq\delta\}\leq
\Tr(\rho_E^2)/(D\delta^2)$.  A union bound over the $K$ pairs proves
Eqs.~\eqref{eq:mixed-union-tail} and \eqref{eq:mixed-max-bound}.
For the stronger trace-distance estimate, define
\begin{equation}
 Y=\sum_{i\neq j}|c_i|^2|c_j|^2|\gamma_{ji}|^2.
\end{equation}
The assumed marginal second moments imply
\begin{equation}
 \E Y=\frac{\Tr(\rho_E^2)}{D}
 \left(1-\sum_i|c_i|^4\right).
\end{equation}
Markov's inequality gives $Y\leq\E Y/\eta$ with probability at least
$1-\eta$.  Substitution into Eq.~\eqref{eq:trace-bound-hs} proves
Eq.~\eqref{eq:mixed-trace-tail}.
\end{proof}

Writing $d_{\rm purity}=1/\Tr(\rho_E^2)$ makes the mean-square scale
\begin{equation}
 \frac{\Tr(\rho_E^2)}{D}=\frac{1}{D d_{\rm purity}}.
\end{equation}
Thus $D$ characterizes the relative-unitary ensemble, whereas
$d_{\rm purity}$ characterizes the initial environmental mixedness.

For a pure environment, Eq.~\eqref{eq:mixedhaar} reduces to $1/D$.  For the
maximally mixed state it gives $1/D^2$.  This stronger suppression does
\emph{not} mean that a maximally mixed environment stores a better record.
The conditional environmental states are
\begin{equation}
 \rho_E^{(i)}=U_i\rho_EU_i^\dagger.
 \label{eq:conditional-density}
\end{equation}
If $\rho_E=I/D$, all of them are identical, so no measurement on the
environment can reveal $i$, even though $\Tr(W_{ji})/D$ may be small.

For pure conditional states, the connection between overlap and record
distinguishability is direct.  With equal priors, the minimum error probability
for discriminating two pure states is
\begin{equation}
 p_{\rm err}
 =\frac12\left(1-\sqrt{1-|\langle E_j|E_i\rangle|^2}\right).
\end{equation}
Equivalently,
\begin{equation}
 D_{\rm tr}\!\left(\ket{E_i}\bra{E_i},\ket{E_j}\bra{E_j}\right)
 =\sqrt{1-|\langle E_j|E_i\rangle|^2}.
 \label{eq:pure-record-distance}
\end{equation}
For mixed states, however, distinguishability is governed by trace distance
or fidelity between the density operators in
Eq.~\eqref{eq:conditional-density}, not by a chosen purification overlap
alone.  In
particular, unitary invariance gives
\begin{equation}
 D_{\rm tr}\!\left(\rho_E^{(i)},\rho_E^{(j)}\right)
 =\frac12\left\|W_{ji}\rho_EW_{ji}^\dagger-\rho_E\right\|_1,
 \label{eq:mixed-record-distance}
\end{equation}
which is generally unrelated to the magnitude of $\Tr(\rho_EW_{ji})$.
Decoherence, locally accessible records, and redundant records are therefore
different properties.  The latter requires information about the pointer
variable to be recoverable from many disjoint environmental
fragments~\cite{blumekohout2006darwinism,zwolak2010nonideal,riedel2016objective}.
Indeed, generic global random states need not have the branching structure
needed for redundant storage~\cite{blumekohout2006darwinism}.  For initially
mixed environments or more general mixed global states, system decoherence
can occur without system--environment entanglement or objective
records~\cite{garciaperez2020decoherence}.  By contrast, for an initially
pure global state, nontrivial reduced-state mixedness entails
system--environment entanglement.

\section{The dynamical bridge}
\label{sec:dynamics}

The preceding propositions are kinematic.  In a microscopic pure-dephasing
model with time-independent conditional environmental Hamiltonians $H_i$,
\begin{equation}
 \gamma_{ji}(t)
 =\bra{E_0}e^{iH_jt}e^{-iH_it}\ket{E_0}.
 \label{eq:echo}
\end{equation}
Its modulus squared is a Loschmidt echo.  Decoherence can consequently be
studied using fidelity decay, spectral correlations, perturbation theory, and
many-body orthogonality effects~\cite{cucchietti2003echo,gorin2004echo}.
Equation~\eqref{eq:echo} also exposes why a large Hilbert space is insufficient:
$e^{iH_jt}e^{-iH_it}$ may remain close to the identity, preserve large
invariant subspaces, or display revivals long before any Poincar\'e recurrence.

Several dimensions should be distinguished.  For an energy window
$[E,E+\Delta E]$, let
\begin{equation}
 D_{\rm shell}=\Tr\,\mathbf 1_{[E,E+\Delta E]}(H_E).
\end{equation}
The Boltzmann entropy associated with that convention is
\begin{equation}
 S_B(E,\Delta E)=k_B\log D_{\rm shell}.
 \label{eq:boltzmann}
\end{equation}
By contrast, let $H_E=\sum_\lambda E_\lambda P_\lambda$ be the spectral
decomposition into distinct energy eigenspaces and put
\begin{equation}
 p_\lambda=\bra{E_0}P_\lambda\ket{E_0},
 \qquad
 d_{\rm spec}=\frac{1}{\sum_\lambda p_\lambda^2}.
 \label{eq:participation}
\end{equation}
This spectral participation ratio measures how broadly the initial state is
distributed over distinct energy eigenspaces and is widely used in
equilibration bounds
\cite{reimann2008foundation,linden2009quantum,short2011equilibration}.
For a nondegenerate Hamiltonian and
$\ket{E_0}=\sum_\alpha a_\alpha\ket{E_\alpha}$, it reduces to
$d_{\rm spec}=1/\sum_\alpha|a_\alpha|^4$.  The projector formulation avoids
dependence on an arbitrary basis chosen inside a degenerate eigenspace.
Neither $D_{\rm shell}$ nor $d_{\rm spec}$ is automatically the dimension
$D$ appearing in a Haar model for $W_{ji}(t)$.

For chaotic many-body systems, eigenstate thermalization and spectral
statistics offer reasons to expect equilibration of appropriate observables,
but they do not imply that a single finite-time orbit is Haar-random over an
entire shell~\cite{rigol2008thermalization,dalessio2016quantum}.  Conserved
charges, locality, finite propagation speed, integrability, many-body
localization, and weak conditional perturbations may all obstruct the Haar
model.  A defensible application therefore requires a statement about the
ensemble or time distribution of the relative unitaries $W_{ji}(t)$.
At a fixed time, fixed Hamiltonians and a fixed initial state determine
$W_{ji}(t)$ deterministically.  Every probability statement must therefore
identify its probability source, such as Hamiltonian disorder, a random
circuit ensemble, random initial states, a distribution of conditional
perturbations, or a random time sampled from a specified window.  A temporal
distribution along one orbit agrees with an ensemble distribution only under
an additional ergodicity or sampling assumption and generally contains
correlated samples.  For time-dependent $H_i(t)$, the exponentials in
Eq.~\eqref{eq:echo} must be replaced by the appropriate time-ordered
propagators.

Approximate unitary designs provide one possible intermediate assumption.
The mean-square identity in Lemma~\ref{lem:mixed} requires only the relevant
second moment, whereas a tail comparable to Eq.~\eqref{eq:tail} requires
stronger moment or concentration control.  Local random circuits are known to
approach approximate designs only after a depth that depends on locality,
system size, design order, and accuracy~\cite{brandao2016local}.  This makes
circuit depth or physical evolution time a meaningful replacement for an
unqualified appeal to Haar typicality.

\section{A local product mechanism}
\label{sec:local-product}

Many decoherence models have a more physical route to small overlaps than
global Haar randomness.  As an idealized mechanism, suppose that the
environment is initially factorized, residual interactions between fragments
can be neglected on the timescale of interest, and the conditional dynamics
preserves the product form
\begin{equation}
 \ket{E_i}=\bigotimes_{k=1}^{N}\ket{e_i^{(k)}}.
\end{equation}
Consequently,
\begin{equation}
\begin{split}
 \langle E_j|E_i\rangle
 =\prod_{k=1}^{N}\langle e_j^{(k)}|e_i^{(k)}\rangle,
 \\
 -\log|\langle E_j|E_i\rangle|
 =\sum_{k=1}^{N}-\log|\langle e_j^{(k)}|e_i^{(k)}\rangle|.
 \label{eq:product-overlap}
\end{split}
\end{equation}
The logarithmic identity assumes nonzero local overlaps; if one local overlap
vanishes, the global overlap vanishes already.
If a positive density of fragments has local overlaps bounded above by a
fixed $q<1$, or more generally if
\begin{equation}
 \liminf_{N\to\infty}\frac1N\sum_{k=1}^N
 \left[-\log|\langle e_j^{(k)}|e_i^{(k)}\rangle|\right]>0,
 \label{eq:positive-overlap-rate}
\end{equation}
Eq.~\eqref{eq:product-overlap} gives exponential suppression in $N$ without
requiring the global state to be Haar-random.  At the same time, the fragment
structure makes the question of redundant records well posed.
This local mechanism and the global geometric estimate are complementary,
not interchangeable.  Approximate rather than exact conditional
factorization produces model-dependent correction terms.

\section{Thermodynamic-limit comparison: infinite tensor products and sectorization}
\label{sec:infinite-tensor-products}

The finite-dimensional estimates above concern generic Haar-distributed
pairs, whose overlaps are nonzero almost surely but typically very small.
Finite-dimensional Hilbert spaces can of course contain exactly orthogonal
states, and a finite product overlap is exactly zero if even one corresponding
pair of local factors is orthogonal.  A related but conceptually distinct
phenomenon occurs when macroscopically different reference sequences define
disjoint representation sectors in an infinite tensor-product limit.

Here $N$ denotes the number of tensor factors, whereas $D$ in the preceding
sections denotes the dimension of the sector on which the random-state or
relative-unitary model is imposed; in many-body applications $D$ may itself
grow exponentially with $N$.  Consider two sequences of normalized local
states,
\begin{equation}
    \ket{E_i^{(N)}}=
    \bigotimes_{k=1}^{N}\ket{e_i^{(k)}},
    \qquad
    \ket{E_j^{(N)}}=
    \bigotimes_{k=1}^{N}\ket{e_j^{(k)}} .
\end{equation}
Their overlap factorizes as
\begin{equation}
    \braket{E_j^{(N)}}{E_i^{(N)}}
    =
    \prod_{k=1}^{N}
    \braket{e_j^{(k)}}{e_i^{(k)}} .
    \label{eq:finite-product-overlap}
\end{equation}
Let $n_N$ be the number of factors among the first $N$ for which
\begin{equation}
    \left|
    \braket{e_j^{(k)}}{e_i^{(k)}}
    \right|
    \leq q<1 ,
\end{equation}
and suppose that $n_N/N\to\alpha>0$.  Then the overlap vanishes exponentially,
\begin{align}
 \left|\braket{E_j^{(N)}}{E_i^{(N)}}\right|
 &\leq q^{n_N}=\exp[-n_N|\log q|]
 \nonumber\\
 &=\exp[-\alpha N|\log q|+o(N)]
 \longrightarrow0,
    \label{eq:thermodynamic-orthogonality}
\end{align}
Thus the sequence of finite-product overlaps tends to zero in the
thermodynamic limit.  This limiting statement alone does not yet establish
inequivalent representations: an infinite tensor-product representation and
an algebra of quasi-local observables must also be specified.

Von Neumann's construction of infinite tensor products makes this limit more
precise~\cite{vonNeumann1939}.  For two normalized reference sequences, a
standard sufficient criterion for their failure to belong to the same weak
equivalence class is
\begin{equation}
 \sum_{k=1}^{\infty}
 \left(1-|\langle e_j^{(k)}|e_i^{(k)}\rangle|\right)=\infty.
 \label{eq:sector-criterion}
\end{equation}
The positive-density condition above implies this divergence.  In the
standard infinite-product construction, such reference sequences can define
orthogonal incomplete tensor-product sectors and, under the corresponding
quasi-local formulation, disjoint or unitarily inequivalent representations
of the observable algebra.  Here ``unitarily inequivalent'' concerns
representations of the physical observable algebra, not the absence of an
abstract isomorphism between Hilbert spaces of the same dimension.

This distinction is important.  At every finite $N$, all product states are
vectors in the same finite-dimensional Hilbert space and can be related by
finite-system unitaries.  Different thermodynamic-limit states may instead
define disjoint representations of the quasi-local algebra.
The latter conclusion is representation dependent and does not follow from
finite-$N$ geometry alone.  It is one mathematical framework used in
thermodynamic-limit approaches to measurement
theory~\cite{hepp-1972,van-den-bossche-2023-a,svozil-2025-u}.

For infinitely many nontrivial finite-dimensional factors, the complete
infinite tensor product is generally nonseparable and decomposes into
mutually orthogonal sectors, while an individual incomplete tensor-product
sector can remain separable.  The complete product may carry continuum many
sector labels.  Writing this cardinal specifically as $\aleph_1$ assumes the
continuum hypothesis.  The physically relevant distinction is the emergence
of disjoint representations of the quasi-local observable algebra, not the
cardinality label by itself.

This infinite-system mechanism may therefore be viewed as an idealized exact
counterpart of the finite-dimensional mechanisms discussed above:
high-dimensional Haar sectors and finite product environments satisfying
Eq.~\eqref{eq:positive-overlap-rate} can yield parametrically small overlaps,
with exponential suppression in $N$ in the product case, whereas a specified
infinite-product construction can yield disjoint sectors.  It does not, by
itself, establish that a physical environment is literally infinite, select
one sector as the unique realized outcome, or derive the Born probabilities.
It does permit macroscopically distinct phases or outcomes to be represented
in inequivalent sectors rather than forcing them into a single irreducible
representation.

\section{Scope of the geometric statement}
\label{sec:scope}

The geometric calculation establishes the following conditional implication:
if a specified finite family of environmental branches is distributed like
independent typical vectors in a $D$-dimensional sector, then its largest
pairwise overlap and the operational distance of the reduced system from its
dephased state are small with high probability.  It also shows that the
projective Hilbert space contains exponentially large pairwise
quasi-orthogonal codes.  When $D=q^N$, composition first creates an
exponentially large ambient dimension and quasi-orthogonal packing then
creates a doubly exponential number of directions at fixed overlap
tolerance.  Simultaneously, concentration compresses most pairwise angular
differences into a narrow neighborhood of $\pi/2$.

It does not establish any of the following claims.
First, it does not select a pointer basis; this depends on the structure and
timescales of the system--environment interaction.  Second, it does not prove
that a physical Hamiltonian generates Haar-typical relative branches.  Third,
it does not turn an improper mixture into a proper ensemble or select a unique
outcome.  Fourth, it does not equate small system coherence with a readable,
let alone redundant, environmental record.  Finally, exponential spherical
packing does not provide exponentially many perfectly decodable classical
labels.

The operational content relevant to Schr\"odinger's ``jellification'' worry is
therefore limited but precise.  For every POVM effect $0\leq M\leq I$,
\begin{equation}
 \left|\Tr\!\left[M\bigl(\rho_S-\Delta(\rho_S)\bigr)\right]\right|
 \leq D_{\rm tr}\!\left(\rho_S,\Delta(\rho_S)\right).
 \label{eq:povm-operational}
\end{equation}
The supremum of the left-hand side over all effects equals the trace distance.
For a Hermitian observable $A$ with $\|A\|_\infty\leq1$, the corresponding
expectation-value difference is bounded by $2D_{\rm tr}$.  Propositions
\ref{prop:simultaneous} and \ref{prop:mixed-simultaneous} therefore directly
bound changes in system measurement probabilities under their respective
typicality assumptions.  They do not assert that the global coherent
superposition has vanished.

\section{Further research}

A first concrete test is to replace the Haar assumption by conditional local
Hamiltonians $H_i=H_E+V_i$.  One can compare chaotic, integrable, and localized
spin environments by measuring the decay, long-time plateau, temporal
variance, and full distribution of Eq.~\eqref{eq:echo}.  Histograms of
$|\gamma_{ij}(t)|^2$ can be compared directly with the
$\operatorname{Beta}(1,D-1)$ law in
Lemma~\ref{lem:haar-overlap}, but the sampling prescription must be stated.
An ensemble histogram and a time histogram along one correlated orbit are not
interchangeable without an ergodicity and sampling analysis.  Their deviation
from the beta law is itself a measure of the failure of pair typicality.

A second problem is finite-time design formation.  For conditional random
circuits or noisy Floquet systems, one can seek explicit bounds of the form
\begin{equation}
 \PP\{|\Tr(\rho_EW_{ji}(t))|^2\geq\eps\}
 \leq f(D,t,\eps,\Tr\rho_E^2),
\end{equation}
interpolating between short-time perturbative behavior and a random-matrix
plateau.  Here $\PP$ must refer to a specified circuit, disorder,
initial-state, perturbation, or random-time ensemble.  Locality should appear
explicitly through circuit depth or a Lieb--Robinson scale.

A third direction is the many-branch problem.  Rather than controlling only
the largest entry, one can study the spectrum of the Gram matrix
$G_{ij}=\langle E_j|E_i\rangle$ and derive bounds tailored to the coefficient
vector $(c_i)$.  This would sharpen Eq.~\eqref{eq:trace-bound-general} and
identify when numerous weak coherences collectively remain observable.

A related geometric problem is to replace Haar measure by physically
generated ensembles while retaining a concentration principle.  One may ask
whether the relevant conditional circuit or Hamiltonian ensemble satisfies a
logarithmic Sobolev inequality, a spectral-gap estimate, or an approximate
design condition strong enough to reproduce a L\'evy-type tail for selected
Lipschitz observables.  On the finite-family side, one can compare the Gram
matrices of dynamically generated branches with Johnson--Lindenstrauss-type
near-isometries.  Such results would connect the general concentration and
embedding principles to realistic, non-Haar dynamics rather than merely
postulating pair typicality.

Finally, for mixed environments one should analyze two quantities in parallel:
the decoherence factor $\Tr(\rho_EW_{ji})$ and the distinguishability or
accessible information of the conditional fragment states
$\rho_F^{(i)}$.  Their separation is necessary for determining when geometric
decoherence is accompanied by objective, redundantly accessible records.

\section{Conclusion}

The central geometric effect is a combination of abundance and
featurelessness.  Tensor composition gives $D=q^N$, exponentially many
linear dimensions in the number of factors.  Relaxing exact orthogonality to
a fixed overlap tolerance then permits $\exp(c_\eps D)$ directions, hence a
doubly exponential number in $N$.  Yet typical overlaps are only
$O(D^{-1/2})$ in amplitude, so almost all projective angles lie within
$O(D^{-1/2})$ of $\pi/2$.  Dimensional capacity therefore explodes while
typical angular structure becomes increasingly homogeneous.

L\'evy's lemma supplies the general reason: regular observables concentrate
on the high-dimensional unit sphere, and the overlap observable concentrates
about the small mean $1/D$.  The exact beta law sharpens this general result
at the natural overlap scale.  Johnson--Lindenstrauss supplies the
complementary finite-configuration statement: logarithmic dimension suffices
to retain the pairwise geometry of exponentially many points.  The direct
Haar packing calculation joins these two perspectives for complex projective
Hilbert space.

This geometry supplies a useful but conditional component of decoherence.
Independent Haar-random pure states have the exact overlap tail
$(1-\eps)^{D-1}$.  For a finite family, this tail yields a
high-probability bound on both the largest branch overlap and the trace
distance between the reduced system and its pointer-basis dephasing.  For a
mixed initial environment, Haar-random relative dynamics instead gives the
mean-square coherence $\Tr(\rho_E^2)/D$ and the simultaneous bounds of
Proposition~\ref{prop:mixed-simultaneous}.  The Hilbert--Schmidt estimate
\eqref{eq:trace-bound-hs} additionally controls the collective contribution
of many weak coherences.  It provides a complementary collective bound and
can be substantially sharper than the entrywise estimate when many weak
coherences are broadly distributed, although it is not uniformly sharper.

The packing number is not an additional Hilbert-space dimension, and angular
homogenization is not literal destruction of global information.  These
results quantify geometric capacity, not dynamics.  The physical question is
whether conditional relative unitaries generated by a particular local
Hamiltonian acquire the necessary typicality on the relevant timescale.
Moreover, loss of system coherence, environmental distinguishability, and
redundant record formation remain separate notions.  With these distinctions
made explicit, high-dimensional geometry provides a clean quantitative bound
on locally accessible interference without being asked to solve the preferred
basis, definite-outcome, or objectivity problems by itself.

\begin{acknowledgments}
This research was funded in whole or in part by the \textit{Austrian Science
Fund (FWF)} [Grant
\href{https://doi.org/10.55776/PIN5424624}{DOI: 10.55776/PIN5424624}].
The author acknowledges TU Wien Bibliothek for financial support through its
Open Access Funding Programme.

OpenAI Codex (GPT-5.6) was used to assist with manuscript criticism, literature
organization, checking derivations, and editorial revision.  The author
directed its use, independently verified the mathematical arguments and cited
sources, revised the resulting text, and assumes full responsibility for the
manuscript.
\end{acknowledgments}

\bibliography{svozil}

\begin{thebibliography}{35}%
\makeatletter
\providecommand \@ifxundefined [1]{%
 \@ifx{#1\undefined}
}%
\providecommand \@ifnum [1]{%
 \ifnum #1\expandafter \@firstoftwo
 \else \expandafter \@secondoftwo
 \fi
}%
\providecommand \@ifx [1]{%
 \ifx #1\expandafter \@firstoftwo
 \else \expandafter \@secondoftwo
 \fi
}%
\providecommand \natexlab [1]{#1}%
\providecommand \enquote  [1]{``#1''}%
\providecommand \bibnamefont  [1]{#1}%
\providecommand \bibfnamefont [1]{#1}%
\providecommand \citenamefont [1]{#1}%
\providecommand \href@noop [0]{\@secondoftwo}%
\providecommand \href [0]{\begingroup \@sanitize@url \@href}%
\providecommand \@href[1]{\@@startlink{#1}\@@href}%
\providecommand \@@href[1]{\endgroup#1\@@endlink}%
\providecommand \@sanitize@url [0]{\catcode `\\12\catcode `\$12\catcode
  `\&12\catcode `\#12\catcode `\^12\catcode `\_12\catcode `\%12\relax}%
\providecommand \@@startlink[1]{}%
\providecommand \@@endlink[0]{}%
\providecommand \url  [0]{\begingroup\@sanitize@url \@url }%
\providecommand \@url [1]{\endgroup\@href {#1}{\urlprefix }}%
\providecommand \urlprefix  [0]{URL }%
\providecommand \Eprint [0]{\href }%
\providecommand \doibase [0]{https://doi.org/}%
\providecommand \selectlanguage [0]{\@gobble}%
\providecommand \bibinfo  [0]{\@secondoftwo}%
\providecommand \bibfield  [0]{\@secondoftwo}%
\providecommand \translation [1]{[#1]}%
\providecommand \BibitemOpen [0]{}%
\providecommand \bibitemStop [0]{}%
\providecommand \bibitemNoStop [0]{.\EOS\space}%
\providecommand \EOS [0]{\spacefactor3000\relax}%
\providecommand \BibitemShut  [1]{\csname bibitem#1\endcsname}%
\let\auto@bib@innerbib\@empty
%</preamble>
\bibitem [{\citenamefont {Gleason}(1957)}]{gleason}%
  \BibitemOpen
  \bibfield  {author} {\bibinfo {author} {\bibfnamefont {A.~M.}\ \bibnamefont
  {Gleason}},\ }\bibfield  {title} {\bibinfo {title} {Measures on the closed
  subspaces of a {H}ilbert space},\ }\href
  {https://doi.org/10.1512/iumj.1957.6.56050} {\bibfield  {journal} {\bibinfo
  {journal} {Journal of Mathematics and Mechanics (now Indiana University
  Mathematics Journal)}\ }\textbf {\bibinfo {volume} {6}},\ \bibinfo {pages}
  {885} (\bibinfo {year} {1957})}\BibitemShut {NoStop}%
\bibitem [{\citenamefont {Wright}\ and\ \citenamefont
  {Weigert}(2019)}]{Wright_2019}%
  \BibitemOpen
  \bibfield  {author} {\bibinfo {author} {\bibfnamefont {V.~J.}\ \bibnamefont
  {Wright}}\ and\ \bibinfo {author} {\bibfnamefont {S.}~\bibnamefont
  {Weigert}},\ }\bibfield  {title} {\bibinfo {title} {A gleason-type theorem
  for qubits based on mixtures of projective measurements},\ }\href
  {https://doi.org/10.1088/1751-8121/aaf93d} {\bibfield  {journal} {\bibinfo
  {journal} {Journal of Physics A: Mathematical and Theoretical}\ }\textbf
  {\bibinfo {volume} {52}},\ \bibinfo {pages} {055301} (\bibinfo {year}
  {2019})}\BibitemShut {NoStop}%
\bibitem [{\citenamefont {{Everett III}}(1957)}]{everett}%
  \BibitemOpen
  \bibfield  {author} {\bibinfo {author} {\bibfnamefont {H.}~\bibnamefont
  {{Everett III}}},\ }\bibfield  {title} {\bibinfo {title} {`{R}elative
  {S}tate' formulation of quantum mechanics},\ }\href
  {https://doi.org/10.1103/RevModPhys.29.454} {\bibfield  {journal} {\bibinfo
  {journal} {Reviews of Modern Physics}\ }\textbf {\bibinfo {volume} {29}},\
  \bibinfo {pages} {454} (\bibinfo {year} {1957})}\BibitemShut {NoStop}%
\bibitem [{\citenamefont {Schlosshauer}(2007)}]{schlosshauer-2007}%
  \BibitemOpen
  \bibfield  {author} {\bibinfo {author} {\bibfnamefont {M.}~\bibnamefont
  {Schlosshauer}},\ }\href {https://doi.org/10.1007/978-3-540-35775-9} {\emph
  {\bibinfo {title} {Decoherence and the Quantum-To-Classical Transition}}},\
  The Frontiers Collection\ (\bibinfo  {publisher} {Springer Verlag},\ \bibinfo
  {address} {Berlin, Heidelberg, Germany},\ \bibinfo {year} {2007})\BibitemShut
  {NoStop}%
\bibitem [{\citenamefont {Schlosshauer}(2019)}]{schlosshauer-2019}%
  \BibitemOpen
  \bibfield  {author} {\bibinfo {author} {\bibfnamefont {M.}~\bibnamefont
  {Schlosshauer}},\ }\bibfield  {title} {\bibinfo {title} {Quantum
  decoherence},\ }\href {https://doi.org/10.1016/j.physrep.2019.10.001}
  {\bibfield  {journal} {\bibinfo  {journal} {Physics Reports}\ }\textbf
  {\bibinfo {volume} {831}},\ \bibinfo {pages} {1} (\bibinfo {year} {2019})},\
  \Eprint {https://arxiv.org/abs/arXiv:1911.06282} {arXiv:1911.06282}
  \BibitemShut {NoStop}%
\bibitem [{\citenamefont {Schr{\"{o}}dinger}(1935)}]{schrodinger-gwsidqm1}%
  \BibitemOpen
  \bibfield  {author} {\bibinfo {author} {\bibfnamefont {E.}~\bibnamefont
  {Schr{\"{o}}dinger}},\ }\bibfield  {title} {\bibinfo {title} {{D}ie
  gegenw\"artige {S}ituation in der {Q}uantenmechanik},\ }\href
  {https://doi.org/10.1007/BF01491891} {\bibfield  {journal} {\bibinfo
  {journal} {Naturwissenschaften}\ }\textbf {\bibinfo {volume} {23}},\ \bibinfo
  {pages} {807} (\bibinfo {year} {1935})}\BibitemShut {NoStop}%
\bibitem [{\citenamefont
  {Schr{\"{o}}dinger}(1995)}]{schroedinger-interpretation}%
  \BibitemOpen
  \bibfield  {author} {\bibinfo {author} {\bibfnamefont {E.}~\bibnamefont
  {Schr{\"{o}}dinger}},\ }\href@noop {} {\emph {\bibinfo {title} {The
  Interpretation of Quantum Mechanics. {D}ublin Seminars (1949-1955) and Other
  Unpublished Essays}}}\ (\bibinfo  {publisher} {Ox Bow Press},\ \bibinfo
  {address} {Woodbridge, Connecticut, USA},\ \bibinfo {year} {1995})\ \bibinfo
  {note} {edited by Michel Bitbol}\BibitemShut {NoStop}%
\bibitem [{\citenamefont {Joos}\ and\ \citenamefont
  {Zeh}(1985)}]{joos1985emergence}%
  \BibitemOpen
  \bibfield  {author} {\bibinfo {author} {\bibfnamefont {E.}~\bibnamefont
  {Joos}}\ and\ \bibinfo {author} {\bibfnamefont {H.~D.}\ \bibnamefont {Zeh}},\
  }\bibfield  {title} {\bibinfo {title} {The emergence of classical properties
  through interaction with the environment},\ }\href
  {https://doi.org/10.1007/BF01725541} {\bibfield  {journal} {\bibinfo
  {journal} {Zeitschrift f{\"u}r Physik B Condensed Matter}\ }\textbf {\bibinfo
  {volume} {59}},\ \bibinfo {pages} {223} (\bibinfo {year} {1985})}\BibitemShut
  {NoStop}%
\bibitem [{\citenamefont {Zurek}(2003)}]{RevModPhys.75.715}%
  \BibitemOpen
  \bibfield  {author} {\bibinfo {author} {\bibfnamefont {W.~H.}\ \bibnamefont
  {Zurek}},\ }\bibfield  {title} {\bibinfo {title} {Decoherence, einselection,
  and the quantum origins of the classical},\ }\href
  {https://doi.org/10.1103/RevModPhys.75.715} {\bibfield  {journal} {\bibinfo
  {journal} {Reviews of Modern Physics}\ }\textbf {\bibinfo {volume} {75}},\
  \bibinfo {pages} {715} (\bibinfo {year} {2003})}\BibitemShut {NoStop}%
\bibitem [{\citenamefont {Joos}\ \emph {et~al.}(2003)\citenamefont {Joos},
  \citenamefont {Zeh}, \citenamefont {Kiefer}, \citenamefont {Giulini},
  \citenamefont {Kupsch},\ and\ \citenamefont
  {Stamatescu}}]{joos2003decoherence}%
  \BibitemOpen
  \bibfield  {author} {\bibinfo {author} {\bibfnamefont {E.}~\bibnamefont
  {Joos}}, \bibinfo {author} {\bibfnamefont {H.~D.}\ \bibnamefont {Zeh}},
  \bibinfo {author} {\bibfnamefont {C.}~\bibnamefont {Kiefer}}, \bibinfo
  {author} {\bibfnamefont {D.}~\bibnamefont {Giulini}}, \bibinfo {author}
  {\bibfnamefont {J.}~\bibnamefont {Kupsch}},\ and\ \bibinfo {author}
  {\bibfnamefont {I.-O.}\ \bibnamefont {Stamatescu}},\ }\href
  {https://doi.org/10.1007/978-3-662-05328-7} {\emph {\bibinfo {title}
  {Decoherence and the Appearance of a Classical World in Quantum Theory}}},\
  \bibinfo {edition} {2nd}\ ed.\ (\bibinfo  {publisher} {Springer},\ \bibinfo
  {address} {Berlin, Heidelberg},\ \bibinfo {year} {2003})\BibitemShut
  {NoStop}%
\bibitem [{\citenamefont {Tao}(2012)}]{tao2012topics}%
  \BibitemOpen
  \bibfield  {author} {\bibinfo {author} {\bibfnamefont {T.}~\bibnamefont
  {Tao}},\ }\href {https://doi.org/10.1090/gsm/132} {\emph {\bibinfo {title}
  {Topics in random matrix theory}}},\ \bibinfo {series} {Graduate Studies in
  Mathematics}, Vol.\ \bibinfo {volume} {132}\ (\bibinfo  {publisher} {American
  Mathematical Society},\ \bibinfo {year} {2012})\BibitemShut {NoStop}%
\bibitem [{\citenamefont {Vershynin}(2018)}]{vershynin2018high}%
  \BibitemOpen
  \bibfield  {author} {\bibinfo {author} {\bibfnamefont {R.}~\bibnamefont
  {Vershynin}},\ }\href {https://doi.org/10.1017/9781108231596} {\emph
  {\bibinfo {title} {High-Dimensional Probability: An Introduction with
  Applications in Data Science}}}\ (\bibinfo  {publisher} {Cambridge University
  Press},\ \bibinfo {year} {2018})\BibitemShut {NoStop}%
\bibitem [{\citenamefont {Ledoux}(2001)}]{ledoux2001concentration}%
  \BibitemOpen
  \bibfield  {author} {\bibinfo {author} {\bibfnamefont {M.}~\bibnamefont
  {Ledoux}},\ }\href {https://doi.org/10.1090/surv/089} {\emph {\bibinfo
  {title} {The Concentration of Measure Phenomenon}}}\ (\bibinfo  {publisher}
  {American Mathematical Society},\ \bibinfo {year} {2001})\BibitemShut
  {NoStop}%
\bibitem [{\citenamefont {Milman}\ and\ \citenamefont
  {Schechtman}(1986)}]{milman1986asymptotic}%
  \BibitemOpen
  \bibfield  {author} {\bibinfo {author} {\bibfnamefont {V.~D.}\ \bibnamefont
  {Milman}}\ and\ \bibinfo {author} {\bibfnamefont {G.}~\bibnamefont
  {Schechtman}},\ }\href {https://doi.org/10.1007/BFb0074900} {\emph {\bibinfo
  {title} {Asymptotic Theory of Finite Dimensional Normed Spaces}}},\ \bibinfo
  {series} {Lecture Notes in Mathematics}, Vol.\ \bibinfo {volume} {1200}\
  (\bibinfo  {publisher} {Springer},\ \bibinfo {address} {Berlin, Heidelberg},\
  \bibinfo {year} {1986})\BibitemShut {NoStop}%
\bibitem [{\citenamefont {Hayden}\ \emph {et~al.}(2006)\citenamefont {Hayden},
  \citenamefont {Leung},\ and\ \citenamefont {Winter}}]{Hayden2006}%
  \BibitemOpen
  \bibfield  {author} {\bibinfo {author} {\bibfnamefont {P.}~\bibnamefont
  {Hayden}}, \bibinfo {author} {\bibfnamefont {D.~W.}\ \bibnamefont {Leung}},\
  and\ \bibinfo {author} {\bibfnamefont {A.}~\bibnamefont {Winter}},\
  }\bibfield  {title} {\bibinfo {title} {Aspects of generic entanglement},\
  }\href {https://doi.org/10.1007/s00220-006-1535-6} {\bibfield  {journal}
  {\bibinfo  {journal} {Communications in Mathematical Physics}\ }\textbf
  {\bibinfo {volume} {265}},\ \bibinfo {pages} {95} (\bibinfo {year} {2006})},\
  \Eprint {https://arxiv.org/abs/quant-ph/0407049} {arXiv:quant-ph/0407049}
  \BibitemShut {NoStop}%
\bibitem [{\citenamefont {Johnson}\ and\ \citenamefont
  {Lindenstrauss}(1984)}]{johnson-1984}%
  \BibitemOpen
  \bibfield  {author} {\bibinfo {author} {\bibfnamefont {W.~B.}\ \bibnamefont
  {Johnson}}\ and\ \bibinfo {author} {\bibfnamefont {J.}~\bibnamefont
  {Lindenstrauss}},\ }\bibinfo {title} {Extensions of {L}ipschitz mappings into
  a {H}ilbert space},\ in\ \href {https://doi.org/10.1090/conm/026/737400}
  {\emph {\bibinfo {booktitle} {Conference on Modern Analysis and
  Probability}}},\ \bibinfo {series} {Contemporary Mathematics}, Vol.~\bibinfo
  {volume} {26}\ (\bibinfo  {publisher} {American Mathematical Society},\
  \bibinfo {year} {1984})\ pp.\ \bibinfo {pages} {189--206}\BibitemShut
  {NoStop}%
\bibitem [{\citenamefont {Welch}(1974)}]{welch1974coherence}%
  \BibitemOpen
  \bibfield  {author} {\bibinfo {author} {\bibfnamefont {L.~R.}\ \bibnamefont
  {Welch}},\ }\bibfield  {title} {\bibinfo {title} {Lower bounds on the maximum
  cross correlation of signals},\ }\href
  {https://doi.org/10.1109/TIT.1974.1055219} {\bibfield  {journal} {\bibinfo
  {journal} {IEEE Transactions on Information Theory}\ }\textbf {\bibinfo
  {volume} {20}},\ \bibinfo {pages} {397} (\bibinfo {year} {1974})}\BibitemShut
  {NoStop}%
\bibitem [{\citenamefont {Holevo}(1973)}]{Holevo1973}%
  \BibitemOpen
  \bibfield  {author} {\bibinfo {author} {\bibfnamefont {A.~S.}\ \bibnamefont
  {Holevo}},\ }\bibfield  {title} {\bibinfo {title} {Bounds for the quantity of
  information transmitted by a quantum communication channel},\ }\href@noop {}
  {\bibfield  {journal} {\bibinfo  {journal} {Problemy Peredachi Informatsii}\
  }\textbf {\bibinfo {volume} {9}},\ \bibinfo {pages} {3} (\bibinfo {year}
  {1973})},\ \bibinfo {note} {english translation in \textit{Problems of
  Information Transmission} 9, 177--183 (1973)}\BibitemShut {NoStop}%
\bibitem [{\citenamefont {Nielsen}\ and\ \citenamefont
  {Chuang}(2010)}]{nielsen-book10}%
  \BibitemOpen
  \bibfield  {author} {\bibinfo {author} {\bibfnamefont {M.~A.}\ \bibnamefont
  {Nielsen}}\ and\ \bibinfo {author} {\bibfnamefont {I.~L.}\ \bibnamefont
  {Chuang}},\ }\href {https://doi.org/10.1017/CBO9780511976667} {\emph
  {\bibinfo {title} {Quantum Computation and Quantum Information}}}\ (\bibinfo
  {publisher} {Cambridge University Press},\ \bibinfo {address} {Cambridge},\
  \bibinfo {year} {2010})\ \bibinfo {note} {10th Anniversary
  Edition}\BibitemShut {NoStop}%
\bibitem [{\citenamefont {Blume-Kohout}\ and\ \citenamefont
  {Zurek}(2006)}]{blumekohout2006darwinism}%
  \BibitemOpen
  \bibfield  {author} {\bibinfo {author} {\bibfnamefont {R.}~\bibnamefont
  {Blume-Kohout}}\ and\ \bibinfo {author} {\bibfnamefont {W.~H.}\ \bibnamefont
  {Zurek}},\ }\bibfield  {title} {\bibinfo {title} {Quantum {Darwinism}:
  Entanglement, branches, and the emergent classicality of redundantly stored
  quantum information},\ }\href {https://doi.org/10.1103/PhysRevA.73.062310}
  {\bibfield  {journal} {\bibinfo  {journal} {Physical Review A}\ }\textbf
  {\bibinfo {volume} {73}},\ \bibinfo {pages} {062310} (\bibinfo {year}
  {2006})},\ \Eprint {https://arxiv.org/abs/quant-ph/0505031}
  {arXiv:quant-ph/0505031} \BibitemShut {NoStop}%
\bibitem [{\citenamefont {Zwolak}\ \emph {et~al.}(2010)\citenamefont {Zwolak},
  \citenamefont {Quan},\ and\ \citenamefont {Zurek}}]{zwolak2010nonideal}%
  \BibitemOpen
  \bibfield  {author} {\bibinfo {author} {\bibfnamefont {M.}~\bibnamefont
  {Zwolak}}, \bibinfo {author} {\bibfnamefont {H.~T.}\ \bibnamefont {Quan}},\
  and\ \bibinfo {author} {\bibfnamefont {W.~H.}\ \bibnamefont {Zurek}},\
  }\bibfield  {title} {\bibinfo {title} {Quantum {Darwinism} in nonideal
  environments},\ }\href {https://doi.org/10.1103/PhysRevA.81.062110}
  {\bibfield  {journal} {\bibinfo  {journal} {Physical Review A}\ }\textbf
  {\bibinfo {volume} {81}},\ \bibinfo {pages} {062110} (\bibinfo {year}
  {2010})},\ \Eprint {https://arxiv.org/abs/0911.4307} {arXiv:0911.4307
  [quant-ph]} \BibitemShut {NoStop}%
\bibitem [{\citenamefont {Riedel}\ \emph {et~al.}(2016)\citenamefont {Riedel},
  \citenamefont {Zurek},\ and\ \citenamefont {Zwolak}}]{riedel2016objective}%
  \BibitemOpen
  \bibfield  {author} {\bibinfo {author} {\bibfnamefont {C.~J.}\ \bibnamefont
  {Riedel}}, \bibinfo {author} {\bibfnamefont {W.~H.}\ \bibnamefont {Zurek}},\
  and\ \bibinfo {author} {\bibfnamefont {M.}~\bibnamefont {Zwolak}},\
  }\bibfield  {title} {\bibinfo {title} {Objective past of a quantum universe:
  Redundant records of consistent histories},\ }\href
  {https://doi.org/10.1103/PhysRevA.93.032126} {\bibfield  {journal} {\bibinfo
  {journal} {Physical Review A}\ }\textbf {\bibinfo {volume} {93}},\ \bibinfo
  {pages} {032126} (\bibinfo {year} {2016})},\ \Eprint
  {https://arxiv.org/abs/1312.0331} {arXiv:1312.0331 [quant-ph]} \BibitemShut
  {NoStop}%
\bibitem [{\citenamefont {Garc{\'i}a-P{\'e}rez}\ \emph
  {et~al.}(2020)\citenamefont {Garc{\'i}a-P{\'e}rez}, \citenamefont {Chisholm},
  \citenamefont {Rossi}, \citenamefont {Palma},\ and\ \citenamefont
  {Maniscalco}}]{garciaperez2020decoherence}%
  \BibitemOpen
  \bibfield  {author} {\bibinfo {author} {\bibfnamefont {G.}~\bibnamefont
  {Garc{\'i}a-P{\'e}rez}}, \bibinfo {author} {\bibfnamefont {D.~A.}\
  \bibnamefont {Chisholm}}, \bibinfo {author} {\bibfnamefont {M.~A.~C.}\
  \bibnamefont {Rossi}}, \bibinfo {author} {\bibfnamefont {G.~M.}\ \bibnamefont
  {Palma}},\ and\ \bibinfo {author} {\bibfnamefont {S.}~\bibnamefont
  {Maniscalco}},\ }\bibfield  {title} {\bibinfo {title} {Decoherence without
  entanglement and quantum {Darwinism}},\ }\href
  {https://doi.org/10.1103/PhysRevResearch.2.012061} {\bibfield  {journal}
  {\bibinfo  {journal} {Physical Review Research}\ }\textbf {\bibinfo {volume}
  {2}},\ \bibinfo {pages} {012061(R)} (\bibinfo {year} {2020})},\ \Eprint
  {https://arxiv.org/abs/1907.12447} {arXiv:1907.12447 [quant-ph]} \BibitemShut
  {NoStop}%
\bibitem [{\citenamefont {Cucchietti}\ \emph {et~al.}(2003)\citenamefont
  {Cucchietti}, \citenamefont {Dalvit}, \citenamefont {Paz},\ and\
  \citenamefont {Zurek}}]{cucchietti2003echo}%
  \BibitemOpen
  \bibfield  {author} {\bibinfo {author} {\bibfnamefont {F.~M.}\ \bibnamefont
  {Cucchietti}}, \bibinfo {author} {\bibfnamefont {D.~A.~R.}\ \bibnamefont
  {Dalvit}}, \bibinfo {author} {\bibfnamefont {J.~P.}\ \bibnamefont {Paz}},\
  and\ \bibinfo {author} {\bibfnamefont {W.~H.}\ \bibnamefont {Zurek}},\
  }\bibfield  {title} {\bibinfo {title} {Decoherence and the {Loschmidt}
  echo},\ }\href {https://doi.org/10.1103/PhysRevLett.91.210403} {\bibfield
  {journal} {\bibinfo  {journal} {Physical Review Letters}\ }\textbf {\bibinfo
  {volume} {91}},\ \bibinfo {pages} {210403} (\bibinfo {year} {2003})},\
  \Eprint {https://arxiv.org/abs/quant-ph/0306142} {arXiv:quant-ph/0306142}
  \BibitemShut {NoStop}%
\bibitem [{\citenamefont {Gorin}\ \emph {et~al.}(2004)\citenamefont {Gorin},
  \citenamefont {Prosen}, \citenamefont {Seligman},\ and\ \citenamefont
  {Strunz}}]{gorin2004echo}%
  \BibitemOpen
  \bibfield  {author} {\bibinfo {author} {\bibfnamefont {T.}~\bibnamefont
  {Gorin}}, \bibinfo {author} {\bibfnamefont {T.}~\bibnamefont {Prosen}},
  \bibinfo {author} {\bibfnamefont {T.~H.}\ \bibnamefont {Seligman}},\ and\
  \bibinfo {author} {\bibfnamefont {W.~T.}\ \bibnamefont {Strunz}},\ }\bibfield
   {title} {\bibinfo {title} {Decoherence alias {Loschmidt} echo of the
  environment},\ }\href {https://doi.org/10.1103/PhysRevA.70.042105} {\bibfield
   {journal} {\bibinfo  {journal} {Physical Review A}\ }\textbf {\bibinfo
  {volume} {70}},\ \bibinfo {pages} {042105} (\bibinfo {year} {2004})},\
  \Eprint {https://arxiv.org/abs/quant-ph/0405011} {arXiv:quant-ph/0405011}
  \BibitemShut {NoStop}%
\bibitem [{\citenamefont {Reimann}(2008)}]{reimann2008foundation}%
  \BibitemOpen
  \bibfield  {author} {\bibinfo {author} {\bibfnamefont {P.}~\bibnamefont
  {Reimann}},\ }\bibfield  {title} {\bibinfo {title} {Foundation of statistical
  mechanics under experimentally realistic conditions},\ }\href
  {https://doi.org/10.1103/PhysRevLett.101.190403} {\bibfield  {journal}
  {\bibinfo  {journal} {Physical Review Letters}\ }\textbf {\bibinfo {volume}
  {101}},\ \bibinfo {pages} {190403} (\bibinfo {year} {2008})}\BibitemShut
  {NoStop}%
\bibitem [{\citenamefont {Linden}\ \emph {et~al.}(2009)\citenamefont {Linden},
  \citenamefont {Popescu}, \citenamefont {Short},\ and\ \citenamefont
  {Winter}}]{linden2009quantum}%
  \BibitemOpen
  \bibfield  {author} {\bibinfo {author} {\bibfnamefont {N.}~\bibnamefont
  {Linden}}, \bibinfo {author} {\bibfnamefont {S.}~\bibnamefont {Popescu}},
  \bibinfo {author} {\bibfnamefont {A.~J.}\ \bibnamefont {Short}},\ and\
  \bibinfo {author} {\bibfnamefont {A.}~\bibnamefont {Winter}},\ }\bibfield
  {title} {\bibinfo {title} {Quantum mechanical evolution towards thermal
  equilibrium},\ }\href {https://doi.org/10.1103/PhysRevE.79.061103} {\bibfield
   {journal} {\bibinfo  {journal} {Physical Review E}\ }\textbf {\bibinfo
  {volume} {79}},\ \bibinfo {pages} {061103} (\bibinfo {year}
  {2009})}\BibitemShut {NoStop}%
\bibitem [{\citenamefont {Short}(2011)}]{short2011equilibration}%
  \BibitemOpen
  \bibfield  {author} {\bibinfo {author} {\bibfnamefont {A.~J.}\ \bibnamefont
  {Short}},\ }\bibfield  {title} {\bibinfo {title} {Equilibration of quantum
  systems and subsystems},\ }\href
  {https://doi.org/10.1088/1367-2630/13/5/053009} {\bibfield  {journal}
  {\bibinfo  {journal} {New Journal of Physics}\ }\textbf {\bibinfo {volume}
  {13}},\ \bibinfo {pages} {053009} (\bibinfo {year} {2011})},\ \Eprint
  {https://arxiv.org/abs/1012.4622} {arXiv:1012.4622 [quant-ph]} \BibitemShut
  {NoStop}%
\bibitem [{\citenamefont {Rigol}\ \emph {et~al.}(2008)\citenamefont {Rigol},
  \citenamefont {Dunjko},\ and\ \citenamefont
  {Olshanii}}]{rigol2008thermalization}%
  \BibitemOpen
  \bibfield  {author} {\bibinfo {author} {\bibfnamefont {M.}~\bibnamefont
  {Rigol}}, \bibinfo {author} {\bibfnamefont {V.}~\bibnamefont {Dunjko}},\ and\
  \bibinfo {author} {\bibfnamefont {M.}~\bibnamefont {Olshanii}},\ }\bibfield
  {title} {\bibinfo {title} {Thermalization and its mechanism for generic
  isolated quantum systems},\ }\href {https://doi.org/10.1038/nature06838}
  {\bibfield  {journal} {\bibinfo  {journal} {Nature}\ }\textbf {\bibinfo
  {volume} {452}},\ \bibinfo {pages} {854} (\bibinfo {year}
  {2008})}\BibitemShut {NoStop}%
\bibitem [{\citenamefont {D'Alessio}\ \emph {et~al.}(2016)\citenamefont
  {D'Alessio}, \citenamefont {Kafri}, \citenamefont {Polkovnikov},\ and\
  \citenamefont {Rigol}}]{dalessio2016quantum}%
  \BibitemOpen
  \bibfield  {author} {\bibinfo {author} {\bibfnamefont {L.}~\bibnamefont
  {D'Alessio}}, \bibinfo {author} {\bibfnamefont {Y.}~\bibnamefont {Kafri}},
  \bibinfo {author} {\bibfnamefont {A.}~\bibnamefont {Polkovnikov}},\ and\
  \bibinfo {author} {\bibfnamefont {M.}~\bibnamefont {Rigol}},\ }\bibfield
  {title} {\bibinfo {title} {From quantum chaos and eigenstate thermalization
  to statistical mechanics and thermodynamics},\ }\href
  {https://doi.org/10.1080/00018732.2016.1198154} {\bibfield  {journal}
  {\bibinfo  {journal} {Advances in Physics}\ }\textbf {\bibinfo {volume}
  {65}},\ \bibinfo {pages} {239} (\bibinfo {year} {2016})}\BibitemShut
  {NoStop}%
\bibitem [{\citenamefont {Brand{\~a}o}\ \emph {et~al.}(2016)\citenamefont
  {Brand{\~a}o}, \citenamefont {Harrow},\ and\ \citenamefont
  {Horodecki}}]{brandao2016local}%
  \BibitemOpen
  \bibfield  {author} {\bibinfo {author} {\bibfnamefont {F.~G. S.~L.}\
  \bibnamefont {Brand{\~a}o}}, \bibinfo {author} {\bibfnamefont {A.~W.}\
  \bibnamefont {Harrow}},\ and\ \bibinfo {author} {\bibfnamefont
  {M.}~\bibnamefont {Horodecki}},\ }\bibfield  {title} {\bibinfo {title} {Local
  random quantum circuits are approximate polynomial-designs},\ }\href
  {https://doi.org/10.1007/s00220-016-2706-8} {\bibfield  {journal} {\bibinfo
  {journal} {Communications in Mathematical Physics}\ }\textbf {\bibinfo
  {volume} {346}},\ \bibinfo {pages} {397} (\bibinfo {year}
  {2016})}\BibitemShut {NoStop}%
\bibitem [{\citenamefont {von Neumann}(1939)}]{vonNeumann1939}%
  \BibitemOpen
  \bibfield  {author} {\bibinfo {author} {\bibfnamefont {J.}~\bibnamefont {von
  Neumann}},\ }\bibfield  {title} {\bibinfo {title} {On infinite direct
  products},\ }\href {http://www.numdam.org/item/CM_1939__6__1_0/} {\bibfield
  {journal} {\bibinfo  {journal} {Compositio Mathematica}\ }\textbf {\bibinfo
  {volume} {6}},\ \bibinfo {pages} {1} (\bibinfo {year} {1939})},\ \bibinfo
  {note} {reprinted in {\sl John {von Neumann}, Collected Works, Vol. III}, A.
  H. Taub, editor, Pergamon Press, New York, 1961, nr. 6, p.
  323--399}\BibitemShut {NoStop}%
\bibitem [{\citenamefont {Hepp}(1972)}]{hepp-1972}%
  \BibitemOpen
  \bibfield  {author} {\bibinfo {author} {\bibfnamefont {K.}~\bibnamefont
  {Hepp}},\ }\bibfield  {title} {\bibinfo {title} {Quantum theory of
  measurement and macroscopic observables},\ }\href
  {https://doi.org/10.5169/seals-114381} {\bibfield  {journal} {\bibinfo
  {journal} {Helvetica Physica Acta}\ }\textbf {\bibinfo {volume} {45}},\
  \bibinfo {pages} {237} (\bibinfo {year} {1972})}\BibitemShut {NoStop}%
\bibitem [{\citenamefont {Van Den~Bossche}\ and\ \citenamefont
  {Grangier}(2023)}]{van-den-bossche-2023-a}%
  \BibitemOpen
  \bibfield  {author} {\bibinfo {author} {\bibfnamefont {M.}~\bibnamefont {Van
  Den~Bossche}}\ and\ \bibinfo {author} {\bibfnamefont {P.}~\bibnamefont
  {Grangier}},\ }\bibfield  {title} {\bibinfo {title} {Contextual unification
  of classical and quantum physics},\ }\href
  {https://doi.org/10.1007/s10701-023-00678-x} {\bibfield  {journal} {\bibinfo
  {journal} {Foundations of Physics}\ }\textbf {\bibinfo {volume} {53}},\
  \bibinfo {pages} {1} (\bibinfo {year} {2023})},\ \Eprint
  {https://arxiv.org/abs/arXiv:2209.01463} {arXiv:2209.01463} \BibitemShut
  {NoStop}%
\bibitem [{\citenamefont {Svozil}(2026)}]{svozil-2025-u}%
  \BibitemOpen
  \bibfield  {author} {\bibinfo {author} {\bibfnamefont {K.}~\bibnamefont
  {Svozil}},\ }\bibfield  {title} {\bibinfo {title} {From unitarity to
  irreversibility: The role of infinite tensor products and nested wigner's
  friends},\ }\href {https://doi.org/10.1007/s10701-025-00903-9} {\bibfield
  {journal} {\bibinfo  {journal} {Foundations of Physics}\ }\textbf {\bibinfo
  {volume} {56}},\ \bibinfo {pages} {4} (\bibinfo {year} {2026})}\BibitemShut
  {NoStop}%
\end{thebibliography}%

\end{document}